\documentclass[conference]{IEEEtran}
\usepackage{blindtext}
\usepackage{graphicx}
\usepackage[utf8]{inputenc}
\usepackage{graphics} 
\usepackage{epsfig} 
\usepackage{mathptmx} 
\usepackage{amsmath} 
\usepackage{amssymb}  
\usepackage{cite}
\usepackage{multicol}
\usepackage{mathtools, cuted}
\usepackage[cmintegrals]{newtxmath}
\usepackage{epstopdf}
\usepackage{graphics} 
\usepackage{epsfig} 

\usepackage{mathptmx} 
\usepackage{amsmath} 
\usepackage{amssymb}  
\usepackage{cite}
\usepackage{multicol}
\usepackage{mathtools, cuted}
\usepackage[cmintegrals]{newtxmath}
\usepackage{mathtools}

\usepackage{mathptmx}
\usepackage{helvet}
\usepackage{courier}
\usepackage{subfigure}
\usepackage{type1cm}
\usepackage{titlesec}
\usepackage{epsfig} 
\usepackage{mathptmx} 
\usepackage{amsmath} 
\usepackage{amssymb}  
\usepackage{cite}
\usepackage{multicol}
\usepackage{mathtools, cuted}
\usepackage[cmintegrals]{newtxmath}
\usepackage{makeidx}         
\usepackage[ruled,vlined]{algorithm2e}
\usepackage{algorithmic}

\usepackage{multicol}        
\usepackage[bottom]{footmisc}
\usepackage{caption}
\usepackage{nomencl}
\usepackage[usenames, dvipsnames]{color}
\usepackage{algorithmic}
\usepackage{multicol}        
\usepackage[bottom]{footmisc}
\usepackage{caption}
\usepackage{titlesec}
\usepackage{nomencl}
\usepackage[usenames, dvipsnames]{color}
\usepackage{multirow}

\usepackage{amsthm}
 
\theoremstyle{definition}

\theoremstyle{theorem}

\theoremstyle{remark}
\newtheorem{remark}{Remark}
\theoremstyle{proposition}
\newtheorem{proposition}{Proposition}
\theoremstyle{corollary}

\theoremstyle{proof}

\newtheorem{assumption}{Assumption}
\theoremstyle{assumption}

\theoremstyle{lemma}
\ifCLASSINFOpdf
\else
\fi
\begin{document}
%
\title{Continuum Deformation Coordination of Multi-Agent Systems Using Cooperative Localization }



%
\author{\IEEEauthorblockN{Hossein Rastgoftar\IEEEauthorrefmark{1},
Sergey Nersesov\IEEEauthorrefmark{1}, and
Hashem Ashrafiuon\IEEEauthorrefmark{1}
}
\IEEEauthorblockA{\IEEEauthorrefmark{1}Mechanical Engineering Department, Villanova University, Villanova, USA 19085}}


\maketitle

\begin{abstract}
This paper studies the problem of decentralized continuum deformation coordination of multi-agent systems aided by cooperative localization. We treat agents as particles inside a triangular continuum (deformable body) in a $2$-D motion space and let the continuum deformation coordination be defined by three leaders located at vertices of a triangle, called the \textit{leading triangle}. The leaders' desired trajectories are assigned as the solution of a constrained optimal control problem such that safety requirements are satisfied in the presence of disturbance and measurement noise.  Followers distributed inside the leading triangle acquire continuum deformation in a decentralized fashion by integrating cooperative localization and local communication. Specifically, cooperative localization estimates the global positions of all agents using relative position measurements based primarily on proximity of agents. Simulation results are presented for a network of ten agents.

\end{abstract}


%
\IEEEpeerreviewmaketitle

\section{Introduction}
Cooperative localization (CL) has shown great promise in reducing state estimation errors for multi-agent systems particularly when no GPS is available \cite{Kia2016, Kia2018}. Cooperative localization algorithms are well-suited in decentralized coordination since they rely only on relative pose measurements and self propagation. Centralized cooperative localization algorithms \cite{Roumeliotis2002, Kia2016} have been shown to be better pose estimators but require each agent keeping track of their cross covariances with all other agents in ad hoc networks. However in networks with predefined communication graphs and/or formation, these same centralized algorithms can be applied in a decentralized manner since each agent only needs to keep track of nearby agents; i.e. those agents that it is able to measure their relative pose using on-board sensors. Hence, we can reliably use these algorithms for continuum deformation coordination of multi-agent systems.

Containment Control and Continuum Deformation Coordination \cite{rastgoftar2016continuum, rastgoftar2020fault} are two existing decentralized leader-follower methods in which a desired coordination is guided by a finite number of leaders and acquired by the remaining followers through local communication. Stability and convergence of the multi-agent containment control method are studied in \cite{liu2015containment, wang2019necessary, ji2008containment}. Multi-agent containment under fixed \cite{li2016containment} and switching \cite{xiong2018containment, wang2018containment} commutation protocols have been investigated. Researchers have also studied finite-time containment control \cite{meng2010distributed, zhao2015finite, he2014finite} as well as multi-agent containment under partial communication \cite{mu2015containment}.
Continuum deformation coordination treats agents as particles of an $n$-D deformable body where the desired coordination is defined by a homogeneous transformation, and $n\in \{1,2,3\}$ is the dimension of the continuum in a $3$-D motion space. A desired $n$-D homogeneous transformation can be defined by $n+1$ leaders agents representing the vertices an $n$-D virtual simplex called \textit{leading simplex}, where the desired trajectories of the leader agents are inferred through local communication. Continuum deformation coordination can formally specify and verify safety by assigning lower limits of eigenvalues of the Jacobian matrix of the homogeneous transformation. As a result, a large number of agents can aggressively deform in an obstacle-laden environment while inter-agent collision avoidance is assured. 

This paper develops a framework for decentralized continuum deformation coordination through simultaneous cooperative localization and local communication. Without loss of generality, this paper assumes that each individual agent is modeled by a double integrator dynamics coordinating in a $2$-D motion space. The desired continuum deformation is planned by the desired trajectories of the three leaders, located at vertices of the leading triangle, and acquired by followers through communication and localization. Assuming that the initial and final configurations of the leading triangles are known, the leaders' desired trajectories are assigned as the solution of a coupled optimization problem. More specifically, the leaders' optimal trajectories are determined as the solution of fixed-time constrained  optimal control problem while an optimization algorithm is employed to minimize travel time between the initial and final formations subject to all safety constraints.

This paper is organized as follows: Preliminary notions of graph theory and a review of homogeneous transformation coordination are presented in Section \ref{Preliminaries}. Problem Statement is presented in Section \ref{Problem Statement} and followed by continuum deformation planning and Cooperative Localization in  Sections \ref{Continuum Deformation Coordination Planning} and \ref{Pose Estimation using Cooperative Localization}, respectively. Safety of continuum deformation coordination is specified in Section  \ref{Safety Specification}. Simulation Results are presented in Section \ref{Simulation Results} followed by Conclusion in Section \ref{Conclusion}.

\section{Preliminaries}\label{Preliminaries}

\subsection{Graph Theory Notions}
\subsubsection{Cooperative Coordination Graph}
Inter-agent communication within the multi-agent system is defined by digraph $\mathcal{G}_c\left(\mathcal{V},\mathcal{E}_c\right)$ with node set $\mathcal{V}$ and edge set $\mathcal{E}_c\subset \mathcal{V}\times \mathcal{V}$. Every agent is identified by a unique index number $i\in \mathcal{V}$. Assuming the multi-agent system (MAS) consists of $N$ agents, $\mathcal{V}=\{1,\cdots,N\}$ can be expressed as $\mathcal{V}=\mathcal{V}_L\bigcup \mathcal{V}_F$, where $\mathcal{V}_L=\{1,\cdots,n+1\}$ and $\mathcal{V}_F=\{n+2,\cdots,N\}$ in and $n$-D continuum deformation coordination, i.e. is the dimension of a homogeneous deformation coordination.

In case when $n=2$, the multi-agent system is treated as particles within a $2$-D continuum. Therefore, $\mathcal{V}_L=\{1,2,3\}$ and $\mathcal{V}_F=\{3,\cdots,N\}$ define the index numbers of leaders and followers, respectively. The leader agents move independently and followers update their positions through inter-agent communication. In particular, when follower $i$ receives information from  three in-neighbor agents, the index numbers of those in-neighbors of the follower $i\in \mathcal{V}_F$ are defined as a set $\mathcal{N}_i=\left\{j\in \mathcal{V}\big|\left(j,i\right)\in \mathcal{E}_c\right\}$. Note that all in-neighbors are the followers contained inside the triangle formed by three leaders. 


In this paper we will make the following assumptions:
\begin{assumption}\label{assm1}
Leaders form a triangle at any time $t$. 
\end{assumption}
\begin{assumption}\label{assm2}
The in-neighbors of every follower $i$ form a triangle at any time $t$.
\end{assumption}
\begin{assumption}\label{assm3}
Every follower $i\in \mathcal{V}$ is inside the communication triangle made by its in-neighbor agents.
\end{assumption}
\begin{assumption}\label{assm4}
The digraph $\mathcal{G}$ is defined such that there exists at least one directed path from every leader to every follower agent.
\end{assumption}
This paper assumes that inner-agent communications have weights and the inter-agent communication topology is time-invariant. Let $w_{i,j}$ denote communication weight between agent $i$ and $j\in \mathcal{N}_i$. Then, we can define weight matrix $\mathbf{W}=\left[W_{ij}\right]\in \mathbb{R}^{\left(N-3\right)\times N}$ as follows:
\begin{equation}
    W_{ij}=
    \begin{cases}
    w_{i+n+1,j}&j\in \mathcal{N}_{i+n+1}\wedge \left(i+n+1\right)\in \mathcal{V}_F\\
    -1&j=i+n+1\\
    0&\mathrm{otherwise}
    \end{cases}
    .
\end{equation}
By partitioning $\mathbf{W}$,
\begin{equation}
\label{partition}
    \mathbf{W}=\begin{bmatrix}
    \mathbf{B}&\mathbf{A}
    \end{bmatrix}
    ,
\end{equation}
it has been proven that $\mathbf{A}\in \mathbb{R}^{\left(N-n-1\right)\times \left(N-n-1\right)}$ and $\mathbf{B}\in \mathbb{R}^{\left(N-n-1\right)\times \left(n+1\right)}$ hold the following properties \cite{rastgoftar2016continuum}:
\begin{enumerate}
    \item{Matrix $\mathbf{A}$ is a nonsingular M-matrix and Hurwitz, if there exists at least one path from every leader to every follower.} 
     \item{Diagonal elements of $\mathbf{A}$ are all $-1$.}
     \item{Matrix $\mathbf{B}$ and off-diagonal elements of $\mathbf{A}$ are non-negative.}
\end{enumerate}
\subsubsection{Cooperative Localization Graph}
We assume that leaders are equipped with GPS and thus no leader agent  needs to estimate its own position. Follower agents rely on cooperative localization to estimate their own positions at any time $t$. Cooperative localization is defined by directed graph $\mathcal{G}_l\left(\mathcal{V},\mathcal{E}_l\right)$ with node set $\mathcal{V}$ and edge set $\mathcal{E}_l\subset \mathcal{V}\times \mathcal{V}$. Note that node sets of the localization and coordination graphs are the same but edge sets $\mathcal{E}_c$ and $\mathcal{E}_l$ are different. 

\subsection{Position Notations}
For every agent $i\in \mathcal{V}$, we define actual position denoted by $\mathbf{r}_i(t)$ at time $t\geq t_0$, global desired position denoted by $\mathbf{r}_{i,HT}(t)$ at time $t\geq t_0$, and reference position denoted by $\mathbf{r}_{i,0}$ at time $t_0$. Note that global actual, global desired, and global reference positions of agent $i\in \mathcal{V}$ are expressed with respect to an inertial coordinate system with base vectors $\hat{\mathbf{e}}_x$ and $\hat{\mathbf{e}}_y$. We define $\hat{\mathbf{e}}_x=\left[1~0\right]^T$ and $\hat{\mathbf{e}}_y=\left[0~1\right]^T$, $\mathbf{r}_i(t)=\left[x_i~y_i\right]^T$, $\mathbf{r}_{i,HT}(t)=\left[x_{i,HT}~y_{i,HT}\right]^T$, and $\mathbf{r}_{i,0}=\left[x_{i,0}~y_{i,0}\right]^T$.

\begin{assumption}\label{assm5}
Global desired position $\mathbf{r}_{i,HT}$ is identical to Global reference position $\mathbf{r}_{i,0}$ at time $t=t_0$ for every agent $i\in \mathcal{V}$, i.e. $\mathbf{r}_{i,HT}(t_0)=\mathbf{r}_{i,0}$ for every agent $i\in \mathcal{V}$.
\end{assumption}




\subsection{Homogeneous Deformation Coordination}
Homogeneous transformation of the multi-agent system is given by
\begin{equation}
\label{HT}
    \forall i\in \mathcal{V},~t\geq t_0,\qquad \mathbf{r}_{i,HT}=\mathbf{Q}(t)\mathbf{r}_{i,0}+\mathbf{d}(t)
\end{equation}
where $\mathbf{Q}(t)\in \mathbb{R}^{2\times 2}$ is non-singular at any time  $t\geq t_0$ and  $\mathbf{Q}(t_0)=\mathbf{I}_2\in \mathbb{R}^{2\times 2}$,  $\mathbf{d}\in \mathbb{R}^{2\times 1}$ is the rigid-body displacement vector. {\color{black}Per Assumption \ref{assm5}, $\mathbf{d}(t_0)=\mathbf{0}$.}

\begin{proposition}
Let $\mathbf{Q}(t)$ be expressed as 
\begin{equation}
t\geq t_0,\qquad     \mathbf{Q}(t)=\mathbf{R}_D(t)\mathbf{U}_D(t)
\end{equation}
using polar decomposiition, where $\mathbf{R}_D(t)$ is an orthogonal (rotation) matrix, and  $\mathbf{U}_D(t)$ is a symmetric (pure deformation) matrix. If $\mathbf{Q}(t_0)=\mathbf{I}_2$ and matrix $\mathbf{Q}(t)$ is non-singular at any time $t$, then, eigenvalues of matrix $\mathbf{U}_D$, denoted by $\lambda_1(t)$ and $\lambda_2(t)$, are all positive at any time $t$ which in turn implies that  matrix $\mathbf{U}_D(t)$ is positive definite at any time $t\geq t_0$. 
\end{proposition}
\begin{proof}
Because $\mathbf{U}_D$ is symmetric, it can be expressed as 
\[
\mathbf{U}_D(t)=\mathbf{S}(t)\mathbf{\Lambda}(t)\mathbf{S}^T(t)
\]
at any time $t\geq t_0$, where $\mathbf{S}(t)$ is orthogonal and $\mathbf{\Lambda}(t)=\mathrm{diag}\left(\lambda_1(t),\lambda_2(t)\right)$ is digonal. 
Now, matrix $\mathbf{Q}(t)$ can be expressed as
\[
\mathbf{Q}(t)=\mathbf{R}_D(t)\mathbf{S}(t)\mathbf{\Lambda}(t)\mathbf{S}^T(t).
\]
Because $\mathbf{R}_D$ and $\mathbf{S}$ are orthogonal and matrix $\mathbf{Q}$ is non-singular at any time $t$, eigenvalues $\lambda_1(t)$ and $\lambda_2(t)$ are non-zero at any time $t$ where  $ \lambda_1(t_0)=\lambda_2(t_0)=1$. Because   $\lambda_1(t_0)$ and $\lambda_2(t_0)$ are positive at time $t=t_0$ and they never become zero due to nonsingularity of matrix $\mathbf{Q}(t)$ at any time $t$, it implies that $\lambda_{1}(t)$ and $\lambda_2(t)$ are positive at any time $t\geq t_0$. Hence, $\mathbf{U}_D(t)$ is positive definite at any time $t\geq t_0$.
\end{proof}

\subsubsection{Homogeneous Deformation Definition}
Since Assumption \ref{assm1} holds at time $t_0$, leaders form a triangle at time $t_0$, and elements of $\mathbf{Q}$ and $\mathbf{d}$ are uniquely defined by the leaders' global desired positions as \cite{rastgoftar2016continuum}
\begin{equation}
    \begin{bmatrix}
    Q_{11}{\color{black}(t)}\\
    Q_{12}{\color{black}(t)}\\
    Q_{21}{\color{black}(t)}\\
    Q_{22}{\color{black}(t)}\\
    d_1{\color{black}(t)}\\
    d_2{\color{black}(t)}
    \end{bmatrix}
    =
    \begin{bmatrix}
    x_{1,0}&y_{1,0}&0&0&1&0\\
    x_{2,0}&y_{2,0}&0&0&1&0\\
    x_{3,0}&y_{3,0}&0&0&1&0\\
    0&0&x_{1,0}&y_{1,0}&0&1\\
    0&0&x_{2,0}&y_{2,0}&0&1\\
    0&0&x_{3,0}&y_{3,0}&0&1\\
    \end{bmatrix}
    ^{-1}
    \begin{bmatrix}
    x_{1,HT}{\color{black}(t)}\\
    x_{2,HT}{\color{black}(t)}\\
    x_{3,HT}{\color{black}(t)}\\
    y_{1,HT}{\color{black}(t)}\\
    y_{2,HT}{\color{black}(t)}\\
    y_{3,HT}{\color{black}(t)}\\
    \end{bmatrix}
    .
\end{equation}
Therefore, a desired homogeneous deformation can be planned either by elements of matrix $\mathbf{Q}$, denoted  by $Q_{11}$, $Q_{12}$, $Q_{21}$, and $Q_{22}$, and vector $\mathbf{d}$, denoted by $d_1$ and $d_2$, or by planning the leaders' global desired position components, denoted by $x_{1,HT}$, $x_{2,HT}$, $x_{3,HT}$, $y_{1,HT}$, $y_{2,HT}$, and $x_{3,HT}$.

Because homogeneous transformation is a linear transformation, global desired position of follower $i\in \mathcal{V}_F$ can be equivalently defined by \eqref{HT} or  expressed as a linear combination of the leaders' global desired position by
\begin{equation}
    \forall i\in \mathcal{V}_F,\forall t\geq t_0.\qquad \qquad \mathbf{r}_{i,HT}(t)=\sum_{j\in \mathcal{V}_L}\alpha_{i,j}\mathbf{r}_{j,HT}(t)
\end{equation}
where $\alpha_{i,1}$, $\alpha_{i,2}$, and $\alpha_{i,3}$ are constants that are uniquely assigned by solving
\begin{equation}
   \begin{bmatrix}
   \alpha_{i,1}\\
   \alpha_{i,2}\\
   \alpha_{i,3}
   \end{bmatrix}
   =
   \begin{bmatrix}
   x_{1,0}&x_{2,0}&x_{3,0}\\
   y_{1,0}&y_{2,0}&y_{3,0}\\
   1&1&1
   \end{bmatrix}
   \begin{bmatrix}
   x_{i,0}\\
   y_{i,0}\\
   1
   \end{bmatrix}
\end{equation}
for every agent $i\in \mathcal{V}_F$.


\section{Problem Statement}\label{Problem Statement}
This paper considers coordination of a double integrator agent team moving in  the $x-y$ plane. Dynamics of agent $i\in \mathcal{V}$ is given by
\begin{equation}\label{eq:states_prop}
\mathbf{x}_{i}(k+1) = \mathbf{A}_i\mathbf{x}_{i}(k)  + \mathbf{B}_i\left[\mathbf{u}_i(k) + \mathbf{\eta}_i(k)\right]
\end{equation}
where $\mathbf{x}_{i}=[\mathbf{r}_i^\top~~\dot{\mathbf{r}}_i^\top]^\top$ and $\mathbf{\eta}_i$ refer to the state, and process noise vectors, respectively. Matrices
\begin{subequations}\label{eq:AiBi}
\begin{equation}
    \mathbf{A}_i =
   \begin{bmatrix}
   1 & 0 & \Delta T & 0 \\ 0 & 1 & 0 & \Delta T \\ 0 & 0 & 1 & 0 \\ 0 & 0 & 0 & 1
   \end{bmatrix}
   ,
\end{equation}
\begin{equation}
    \mathbf{B}_i =
   \begin{bmatrix}
   0 & 0 \\ 0 & 0 \\ \Delta T & 0 \\ 0 & \Delta T
   \end{bmatrix}
\end{equation}
\end{subequations}
where $\Delta T$ is the sample time step.
For agent $i\in \mathcal{V}$, we define control input $\mathbf{u}_i$ as follows:
\begin{equation}
\label{main2ninput}
    \mathbf{u}_i=
    \begin{cases}
    g_1\left(\dot{\mathbf{r}}_{i,HT}-\dot{\mathbf{r}}_i\right)+g_2\left({\mathbf{r}}_{i,HT}-{\mathbf{r}}_i\right)&i\in \mathcal{V}_L\\
     g_1\sum_{j\in \mathcal{N}_i}w_{i,j}\left(\dot{\hat{\mathbf{r}}}_{j}-\dot{\hat{\mathbf{r}}}_{i}\right)+g_2\sum_{j\in \mathcal{N}_i}w_{i,j}\left({\hat{\mathbf{r}}}_{j}-{\hat{\mathbf{r}}}_{i}\right)&i\in \mathcal{V}_F\\
    \end{cases}
\end{equation}
where $w_{i,j}>0$ is a constant communication weight between agent $i\in \mathcal{V}_F$ and in-neighbor $j\in \mathcal{N}_i$, {\color{black}$\mathbf{r}_i$ and $\mathbf{r}_j$ denote actual positions of agents $i$ and $j$, $\hat{\mathbf{r}}_i$ and $\hat{\mathbf{r}}_j$ denote the estimations of actual positions of agents $i$ and $j$, and $\hat{\mathbf{r}}_{i,HT}$ is the global desired position of leader agent $i\in \mathcal{V}_L$}. This paper assumes that followers' communication weights are consistent with agents' reference positions and obtained by
\begin{equation}
   \begin{bmatrix}
   w_{i,i_1}\\
   w_{i,i_2}\\
  w_{i,i_3}
   \end{bmatrix}
   =
   \begin{bmatrix}
   x_{i_1,0}&x_{i_2,0}&x_{i_3,0}\\
   y_{i_1,0}&y_{i_2,0}&y_{i_3,0}\\
   1&1&1
   \end{bmatrix}
   \begin{bmatrix}
   x_{i,0}\\
   y_{i,0}\\
   1
   \end{bmatrix}
\end{equation}
where $g_1,g_2>0$ are constant; $i_1$, $i_2$, and $i_3$ are the index numbers of the in-neighbors of follower $i\in \mathcal{V}_F$, i.e. $\mathcal{N}_i=\{i_1,i_2,i_3\}$ defines the index numbers of the in-neighbors of agent $i\in \mathcal{V}_F$.

The above continuum deformation coordination problem is defined as a decentralized leader-follower coordination problem.  In Section \ref{Continuum Deformation Coordination Planning}, we assume that initial and final configurations of leaders are given, and assign the  desired trajectories as a solution of a constrained optimal control problem. We offer a cooperative localization method in Section \ref{Pose Estimation using Cooperative Localization} to acquire a desired continuum deformation coordination in a decentralized fashion. Furthermore, we provide safety conditions in Section \ref{Safety Specification} to check and ensure that collision is avoided in decentralized continuum coordination, inferred by cooperative localization.

\section{Continuum Deformation Coordination Planning} \label{Continuum Deformation Coordination Planning}
In this section, we discuss how the desired positions of the leaders are planned. Specifically, we consider leaders' desired positions planned according to the minimum control effort with fixed initial and final positions and velocities and the fixed area of the triangle made up by three leaders. Indeed, let $\mathbf{r}_{i,HT}$ be updated by the double integrator dynamics
\[
    i\in \mathcal{V},~q\in \{x,y\},\qquad \ddot{q}_{i,HT}=v_{q,i}
\]
It is desired to assign $v_{q,i}$ such that cost function
\begin{equation}
    \mathrm{J}=\int_{t_0}^{t_f}\left(\sum_{i\in \mathcal{V}_L}\left(v_{x,i}^2+v_{y,i}^2\right)\right) \,dt
\end{equation}
is minimized subject to boundary conditions
\begin{subequations}
\label{bcond}
\begin{equation}
q\in \{x,y\},~i\in \mathcal{V}_L,\qquad    q_{i,HT}\left(t_0\right)=q_{i,0},
\end{equation}
\begin{equation}
q\in \{x,y\},~i\in \mathcal{V}_L,\qquad     q_{i,HT}\left(t_f\right)=q_{i,f},
\end{equation}
\begin{equation}
 q\in \{x,y\},~i\in \mathcal{V}_L,\qquad    \dot{q}_{i,HT}\left(t_0\right)=\dot{q}_{i,0},
\end{equation}
\begin{equation}
q\in \{x,y\},~i\in \mathcal{V}_L,\qquad     \dot{q}_{i,HT}\left(t_f\right)=\dot{q}_{i,f},
\end{equation}
\end{subequations}
and equality constraint 
\begin{equation}
    2a_0-\left|
    \begin{array}{ccc}
         x_{1,HT}(t)&x_{2,HT}(t) &x_{3,HT}(t) \\
        y_{1,HT}(t)&y_{2,HT}(t) &y_{3,HT}(t) \\
         1&1&1
    \end{array}
    \right|=0
\end{equation}
where $a_0$ is the area of the leading triangle with vertices occupied by leaders $1$, $2$, and $3$, respectively. Note that the area $a_0$ and travel time $T=t_f-t_0$ are both fixed. The solution of the above optimal control problem was presented in Ref. \cite{rastgoftar2018cooperative}.

\section{Cooperative Localization}\label{Pose Estimation using Cooperative Localization}
In order to control a multi-agent system, at least a position feedback from all agents is needed. However, under practical circumstances such feedback data might not be available. Alternatively, what might be available is the limited information about relative pose of the agents with respect to each other. In this section, we employ cooperative localization for state estimation in order to ensure availability of accurate feedback in the multi-agent system. Specifically, we present an algorithm to reconstruct the full state of each agent based on the limited information about each agent's relative pose.

\subsection{Dynamics of Individual Agents}
Utilizing Extended Kalman Filter (EKF), we denote the predicted and updated state estimates for agent $i$ as $\hat{\mathbf{x}}_{i-}$ and $\hat{\mathbf{x}}_{i+}$, respectively. Referring to \eqref{eq:states_prop}-\eqref{eq:AiBi}, agents $i=1,\cdots,n$ can propagate their states from sample time $(k)$ to $k+1$ given initial conditions ${\mathbf{x}}_i(0)$ and starting with cross covariance $\mathbf{P}_{ij}(0)=0_4$ and initial guess of $\mathbf{P}_{ii}(0)>0$ for covariance,  as follows:
\begin{subequations}\label{eq:propagation}
\begin{align}
\hat{\mathbf{x}}_{i-}(k+1) &=  \mathbf{A}_i\hat{\mathbf{x}}_{i+}(k) + \mathbf{B}_i \mathbf{u}_i(k) \\
\mathbf{P}_{ii-}(k+1) &= \mathbf{B}_i \mathbf{P}_{ii+}(k) \mathbf{B}_i^T + \mathbf{B}_i \mathbf{Q}_i \mathbf{B}_i^T \\
\mathbf{P}_{ij-}(k+1) &= \mathbf{B}_i \mathbf{P}_{ii+}(k) \mathbf{B}_j^T
\end{align}
\end{subequations}
where $j=1,\cdots,n,~j\ne i$, and $\mathbf{Q}_i>0$ is the the process noise covariance.

The key to success of cooperative localization is the cross covariance \cite{Kia2016} which is the correlation term between agent $i$ and $j$ and comes into play when agent $i$ takes relative pose measurements of agent $j$. Next, assume agent $i$ in the set $\{1,\cdots,N\}$ can take a measurement of agent $j\ne i$. Then, the measurement vector ${\mathbf{y}}_{ij}$ can be written as
\begin{equation}\label{eq:output}
{\mathbf{y}}_{ij}(k+1) = {\mathbf{c}}_{ij}({\mathbf{x}}_{i}(k),{\mathbf{x}}_{j}(k)) + \mathbf{\nu}_{ij}(k)
\end{equation}
where $\mathbf{\nu}_{ij}(k)$ represents measurement noise. Assuming agent $i$ has an on-board sensor such as Lidar, then the measurements consist of relative range $d_{ij}$ and azimuth angle $\theta_{ij}$:
\begin{equation}\label{eq:measure}
{\mathbf{c}}_{ij} = \begin{bmatrix} d_{ij} \\ \theta_{ij} \end{bmatrix} = 
\begin{bmatrix} \sqrt{(x_j-x_i)^2+(y_j-y_i)^2} \\ \arctan{\frac{y_j-y_i}{x_j-x_i}} \end{bmatrix}
\end{equation}

At the update step, the state and covariance matrices remain the same if there are no relative measurements. Let us now assume that a follower agent $f$ takes relative measurements of a leader agent $l$. Then, the innovation residual error between the measurement $\mathbf{y}_{fl}$ and the estimated output through propagated states is
\begin{equation}\label{eq:innovation_update}
\begin{split}
&{\mathbf{e}}_{fl}(k+1) = {\mathbf{y}}_{fl}(k+1) - {\mathbf{c}}_{fl}(\hat{\mathbf{x}}_{f-}(k+1),\hat{\mathbf{x}}_{l-}(k+1)) \\
 &\approx {\mathbf{y}}_{fl}(k+1) - \left[\mathbf{C}_{fl}(k+1)~~-\mathbf{C}_{fl}(k+1)\right]
 \begin{bmatrix} \hat{\mathbf{x}}_{f-}(k+1) \\ \hat{\mathbf{x}}_{l-}(k+1) \end{bmatrix}
\end{split}
\end{equation}
where
\begin{equation}\label{eq:H}
\begin{split}
 \mathbf{C}_{fl} = \begin{bmatrix} 
        -\frac{\hat{x}_l-\hat{x}_f}{d_{fl}} & -\frac{\hat{y}_l-\hat{y}_f}{d_{fl}}  & 0 & 0 \\
        \frac{\theta_{fl}}{\hat{x}_l-\hat{x}_f} & -\frac{1}{\hat{x}_l-\hat{x}_f} & 0 & 0
       \end{bmatrix}.
\end{split}
\end{equation}

The innovation covariance matrix is modified due to the cross covariance terms as follows:
\begin{equation}\label{eq:innovation_covariance_update}
\begin{split}
\mathbf{S}_{fl} &= \mathbf{C}_{fl}(k+1) \mathbf{P}_{ff-}(k+1){\mathbf{C}_{fl}}^{T}(k+1) \\
&+ \mathbf{C}_{fl}(k+1)\mathbf{P}_{ll-}(k+1){\mathbf{C}_{fl}}^{T}(k+1) \\
&-2\mathbf{C}_{fl}(k+1) \mathbf{P}_{fl-}(k+1){\mathbf{C}_{fl}}^{T}(k+1) + \mathbf{R}_{fl}(k+1)
\end{split}
\end{equation}
$\mathbf{R}_{fl}>0$ is the the measurement noise covariance. Using (\ref{eq:innovation_update}) and (\ref{eq:innovation_covariance_update}), the Kalman gains, states and covariances for all agents $i=1,\cdots,n$ are updated as follows 
\begin{equation}\label{eq:states_update_estimate}
\begin{split}
\mathbf{K}_i(k+1) &= [\mathbf{P}_{if-}(k+1) - \mathbf{P}_{il-}(k+1)] {\mathbf{C}_{fl}}^T \mathbf{S}_{fl}^{-1} \\
\hat{\mathbf{x}}_{i+}(k+1) &= \hat{\mathbf{x}}_{i-}(k+1) + \mathbf{K}_i (k+1) {\mathbf{e}}_{fl} (k+1) \\
\mathbf{P}_{ii+}(k+1) &= \mathbf{P}_{ii-}(k+1) - \mathbf{K}_i(k+1)\mathbf{S}_{fl}(k+1) {\mathbf{K}_i}^T(k+1) \\
\mathbf{P}_{ij+}(k+1) &= \mathbf{P}_{ij-}(k+1) - \mathbf{K}_i(k+1)\mathbf{S}_{fl}(k+1) {\mathbf{K}_j}^T(k+1)
\end{split}
\end{equation}
where $j=1,\cdots,n,~j\ne i$. 
Note that, the presented method is centralized in ad hoc networks. However, continuum deformation relies on each follower agent following three leaders. Therefore, each agent $i$ only needs to keep track of only three cross covariances with its leaders. This allows this simple cooperative localization algorithm to become effectively decentralized.

To prove observer stability of cooperative localization (CL) algorithm, we define the collection of all estimated states as ${\hat{\mathbf{X}}}_{\mathrm{CL}} = \begin{bmatrix} \hat{\mathbf{x}}_4&\cdots&\hat{\mathbf{x}}_N \end{bmatrix}^T$ and collective input as $\mathbf{U}_{\mathrm{CL}} = \begin{bmatrix} \mathbf{u}_4&\cdots&\mathbf{u}_N \end{bmatrix}^T$. Then, the collective propagation equations are written as:
\begin{subequations}\label{eq:propagation_collective}
\begin{align}
{\hat{\mathbf{X}}}_{\mathrm{CL-}}(k+1) &=  \mathbf{A}_{\mathrm{CL}}{\hat{\mathbf{X}}}_{\mathrm{CL+}}(k) + \mathbf{B}_{\mathrm{CL}} \mathbf{U}_{\mathrm{CL}}(k) \\
\mathbf{P}_{\mathrm{CL-}}(k+1) &= \mathbf{B}_{\mathrm{CL}}^\top \mathbf{P}_{\mathrm{CL+}}(k) \mathbf{B}_{\mathrm{CL}} + \mathbf{B}_{\mathrm{CL}}^T \mathbf{Q}_{\mathrm{CL}} \mathbf{B}_{\mathrm{CL}}
\end{align}
\end{subequations}
where $\mathbf{A}_{\mathrm{CL}} = \mathrm{Diag}\begin{bmatrix} \mathbf{A}_4&\cdots&\mathbf{A}_N \end{bmatrix}\in\mathbb{R}^{4(N-3)\times4(N-3)}$, $\mathbf{B}_{\mathrm{CL}} = \mathrm{Diag}\begin{bmatrix} \mathbf{B}_4&\cdots&\mathbf{B}_N \end{bmatrix}\in\mathbb{R}^{4(N-3)\times2(N-3)}$, $\mathbf{Q}_{\mathrm{CL}} = \mathrm{Diag}\begin{bmatrix} \mathbf{Q}_4&\cdots&\mathbf{Q}_N \end{bmatrix}\in\mathbb{R}^{2(N-3)\times2(N-3)}$, and $\mathbf{P}_{\mathrm{CL}}\in\mathbb{R}^{4(N-3)\times4(N-3)}$ is constructed according to (\ref{eq:propagation}).
The collective update equations are:
\begin{subequations}\label{eq:update_collective}
\begin{align}
\mathbf{K}_{\mathrm{CL}}(k+1) &= \mathbf{P}_{\mathrm{CL-}}(k+1)\mathbf{C}_{\mathrm{CL}}^T(k+1)\mathbf{S}_{\mathrm{CL}}^{-1} \\
{\hat{\mathbf{X}}}_{\mathrm{CL+}}(k+1) &= {\hat{\mathbf{X}}}_{\mathrm{CL-}}(k+1) + \mathbf{K}_{\mathrm{CL}}(k+1)\mathbf{Y}_{CL}\\
\mathbf{P}_{\mathrm{CL+}}(k+1) &= \mathbf{P}_{\mathrm{CL-}}(k+1) - \mathbf{K}_{\mathrm{CL}}(k+1)\mathbf{S}_{\mathrm{CL}}\mathbf{K}_{\mathrm{CL}}^T(k+1)
\end{align}
\end{subequations}
where $\mathbf{Y}_{\mathrm{CL}} = \begin{bmatrix} \mathbf{y}_{fl_1}&\cdots&\mathbf{y}_{fl_M} \end{bmatrix}^T\in\mathbb{R}^{2M\times1}$, $\mathbf{S}_{\mathrm{CL}} = \mathrm{Diag}\begin{bmatrix} \mathbf{S}_{fl_1}&\cdots&\mathbf{S}_{fl_M} \end{bmatrix}\in\mathbb{R}^{2M\times2M}$, and $M$ is the number of measurements of all followers of their leaders. The output matrix $\mathbf{C}_{\mathrm{CL}}\in\mathbb{R}^{2M\times4(N-3)}$, is constructed by assembling at $\mathbf{C}_{fl}$ and $-\mathbf{C}_{fl}$ at columns corresponding to the states of follower $f$ and leader $l$, respectively, for each measurement. Finally, to ensure that the observer is stable, the matrix $\mathbf{A}_{\mathrm{OBS}}(k)=\mathbf{A}_{\mathrm{CL}}-\mathbf{K}_{\mathrm{CL}}(k)\mathbf{C}_{\mathrm{CL}}(k)$ must be such that there exists $P\geq 0$ satisfying
\begin{equation}\label{LyapunovDT}
\mathbf{A}^{T}_{\mathrm{OBS}}(k)P\mathbf{A}_{\mathrm{OBS}}(k)+R(k)=P,
\end{equation}
for each $k\geq 0$, where $R(k)\geq r_{1}I_{4(N-3)}$, $r_{1}>0$. For practical implementation, we set $P=I_{4(N-3)}$ in equation (\ref{LyapunovDT}) and ensure that $\sigma_{\mathrm{max}}(\mathbf{A}_{\mathrm{OBS}}(k))<r_{2}$, where $r_{2}\in(0,1)$, in order to guarantee observer stability. 


\subsection{MAS Collective Dynamics}
Define 
\begin{subequations}
\begin{equation}
    \mathbf{X}_{\mathrm{SYS}}=\mathrm{vec}\left(
    \begin{bmatrix}
    \mathbf{r}_4&\cdots&\mathbf{r}_N&\dot{\mathbf{r}}_4&\cdots&\dot{\mathbf{r}}_N
    \end{bmatrix}
    ^T
    \right)
\end{equation}
\begin{equation}
    \hat{\mathbf{X}}_{\mathrm{SYS}}=\mathrm{vec}\left(
    \begin{bmatrix}
    \hat{\mathbf{r}}_4&\cdots&\hat{\mathbf{r}}_N&\dot{\hat{\mathbf{r}}}_4&\cdots&\dot{\hat{\mathbf{r}}}_N
    \end{bmatrix}
    ^T
    \right)
\end{equation}
\end{subequations}
as the state vectors of the MAS control system and the estimator, respectively, and matrix
\begin{equation}
    \mathbf{K}=\begin{bmatrix}
    \mathbf{0}&\mathbf{0}\\
    g_2\mathbf{A}&g_1\mathbf{A}
    \end{bmatrix}
\end{equation}
where $\mathbf{A}$ was previously defined in \eqref{partition}. Note that $\mathrm{vec}\left(\cdot\right)$ is the matrix vectorization operator. 
Given agent dynamics  \eqref{eq:states_prop}, control inputs \eqref{main2ninput}, and measurement/output equations (\ref{eq:measure}-\ref{eq:output}), the MAS collective dynamics and the estimator dynamics become
\begin{subequations}
\label{canonical}
\begin{equation}
\label{canonicala}
\resizebox{0.99\hsize}{!}{%
$
\begin{cases}
    {\mathbf{X}}_{\mathrm{SYS}}\left(k+1\right)=&\mathbf{A}_{\mathrm{SYS}}{\mathbf{X}}_{\mathrm{SYS}}\left(k\right)+\mathbf{B}_{\mathrm{SYS}}[{\mathbf{U}}_{\mathrm{SYS}}\left(k\right)+\mathbf{\eta}_{\mathrm{SYS}}\left(k\right)]+\mathbf{V}_{\mathrm{SYS}}\left(k\right)\\
    \mathbf{Y}_{\mathrm{SYS}}\left(k\right)=&\mathbf{C}_{\mathrm{SYS}}{\mathbf{X}}_{\mathrm{SYS}}\left(k\right)+\mathbf{\nu}_{\mathrm{SYS}}\left(k\right)
\end{cases}
,
$
}
\end{equation}
\begin{equation}
\label{canonicalb}
\resizebox{0.99\hsize}{!}{%
$
\begin{cases}
    \hat{\mathbf{X}}_{\mathrm{SYS}}\left(k+1\right)=&\mathbf{A}_{\mathrm{SYS}}\hat{{\mathbf{X}}}_{\mathrm{SYS}}\left(k\right)+\mathbf{B}_{\mathrm{SYS}}[{\mathbf{U}}_{\mathrm{SYS}}\left(k\right)+\mathbf{\eta}_{\mathrm{SYS}}\left(k\right)]+\mathbf{V}_{\mathrm{SYS}}\left(k\right)\\
    \hat{\mathbf{Y}}_{\mathrm{SYS}}\left(k\right)=&\mathbf{C}_{\mathrm{SYS}}\hat{{\mathbf{X}}}_{\mathrm{SYS}}\left(k\right)+\mathbf{\nu}_{\mathrm{SYS}}\left(k\right)
\end{cases}
,
$
}
\end{equation}
\end{subequations}
where $\mathbf{C}_{\mathrm{SYS}}\in \mathbb{R}^{2\left(N-3\right)\times 4\left(N-3\right)}$ is the observation matrix, $\mathbf{\eta}_{\mathrm{SYS}}\in \mathbb{R}^{4\left(N-3\right)\times 1}$ and $\mathbf{\nu}_{\mathrm{SYS}}\in \mathbb{R}^{4\left(N-3\right)\times 1}$ are Gaussian process and measurement noise vectors, respectively, and
\begin{subequations}
\begin{equation}
    \mathbf{A}_{\mathrm{SYS}}=\left(\mathbf{I}_{(N-3)}\otimes\begin{bmatrix}
    \mathbf{I}&\Delta T\mathbf{I}\\
    \mathbf{0}&\mathbf{I}
    \end{bmatrix}
    \right)\in \mathbb{R}^{4\left(N-3\right)\times 4\left(N-3\right)},
\end{equation}
\begin{equation}
    \mathbf{B}_{\mathrm{SYS}}=\left(\mathbf{I}_{(N-3)}\otimes\begin{bmatrix}
    \mathbf{0}&\mathbf{0}\\
    g_2\Delta T\mathbf{B}&g_1\Delta T\mathbf{B}
    \end{bmatrix}
    \right)\in \mathbb{R}^{4\left(N-3\right)\times 6},
\end{equation}
\begin{equation}
    \mathbf{V}_{\mathrm{SYS}}=\mathbf{K}\hat{\mathbf{X}}_{\mathrm{SYS}},
\end{equation}
\begin{equation}
    \mathbf{U}_{\mathrm{SYS}}=\mathrm{vec}\left(
    \begin{bmatrix}
    \mathbf{r}_1&\mathbf{r}_2&\mathbf{r}_3&\dot{\mathbf{r}}_1&\dot{\mathbf{r}}_2&\dot{\mathbf{r}}_3
    \end{bmatrix}
    ^T
    \right)\in \mathbb{R}^{6\times 1}.
\end{equation}
\end{subequations}
Note that control $g_1$ and $g_2$ are chosen such that the eigenvalues of matrix $\mathbf{A}_{\mathrm{SYS}}$ are strictly located on the left side of the complex plane.
The block diagram of the controllable form of MAS collective dynamics is shown in Fig. \ref{blockdiagram}.
\begin{figure}[ht]
\centering
\includegraphics[width=3.3   in]{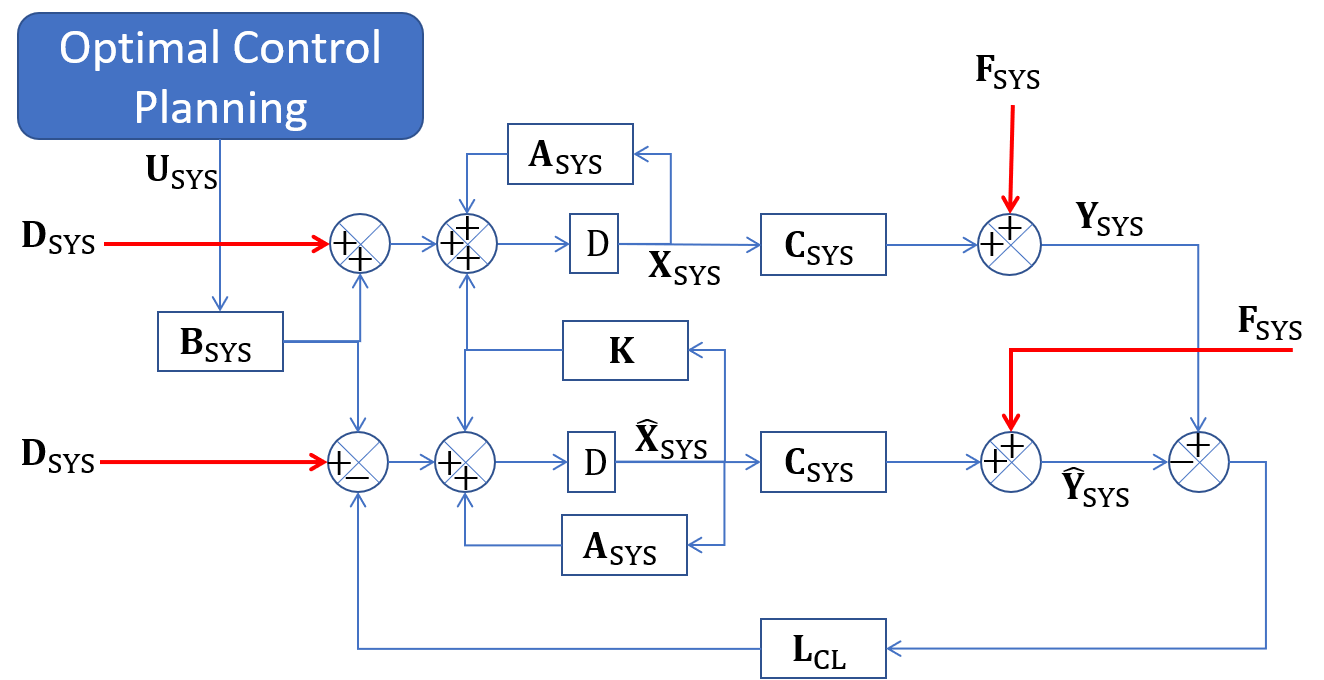}
\caption{Block diagram for the controllable form of the MAS collective dynamics.
}
\label{blockdiagram}
\end{figure}
\begin{remark}
To implement the proposed cooperative localization, we define  the similarity transformations
\begin{subequations}
\begin{equation}
\label{st}
    \mathbf{X}_{\mathrm{CL}}=\mathbf{O}\mathbf{X}_{\mathrm{SYS}},
\end{equation}
\begin{equation}
\label{st}
    \mathbf{A}_{\mathrm{CL}}=\mathbf{O}^T\mathbf{A}_{\mathrm{SYS}}\mathbf{O}
\end{equation}
\begin{equation}
\label{st}
    \mathbf{B}_{\mathrm{CL}}=\mathbf{O}^T\mathbf{B}_{\mathrm{SYS}}
\end{equation}
\end{subequations}
for relating $\mathbf{X}_{\mathrm{SYS}}$,  $\mathbf{A}_{\mathrm{SYS}}$, and $\mathbf{B}_{\mathrm{SYS}}$ to $\mathbf{X}_{\mathrm{CL}}$,  $\mathbf{A}_{\mathrm{CL}}$, and $\mathbf{B}_{\mathrm{CL}}$, respectively, where
$\mathbf{O}=\left[O_{lh}\right]\in \mathbb{R}^{4\left(N-3\right)\times 4\left(N-3\right)}$ is orthonormal and defined as follows:
\begin{equation}
    O_{lh}=\begin{cases}
        1&\mathrm{if~}l=4(i-1)+1,~h=i,~i=4,\cdots,N\\
        1&\mathrm{if~}l=4(i-1)+2,~h=i+N-3,~i=4,\cdots,N\\
        1&\mathrm{if~}l=4(i-1)+3,~h=i+2(N-3),~i=4,\cdots,N\\
        1&\mathrm{if~}l=4(i-1)+4,~h=i+3(N-3),~i=4,\cdots,N\\
        0&\mathrm{otherwise}
    \end{cases}
\end{equation}
\end{remark}

\begin{figure}[ht]
\centering
\includegraphics[width=3.3   in]{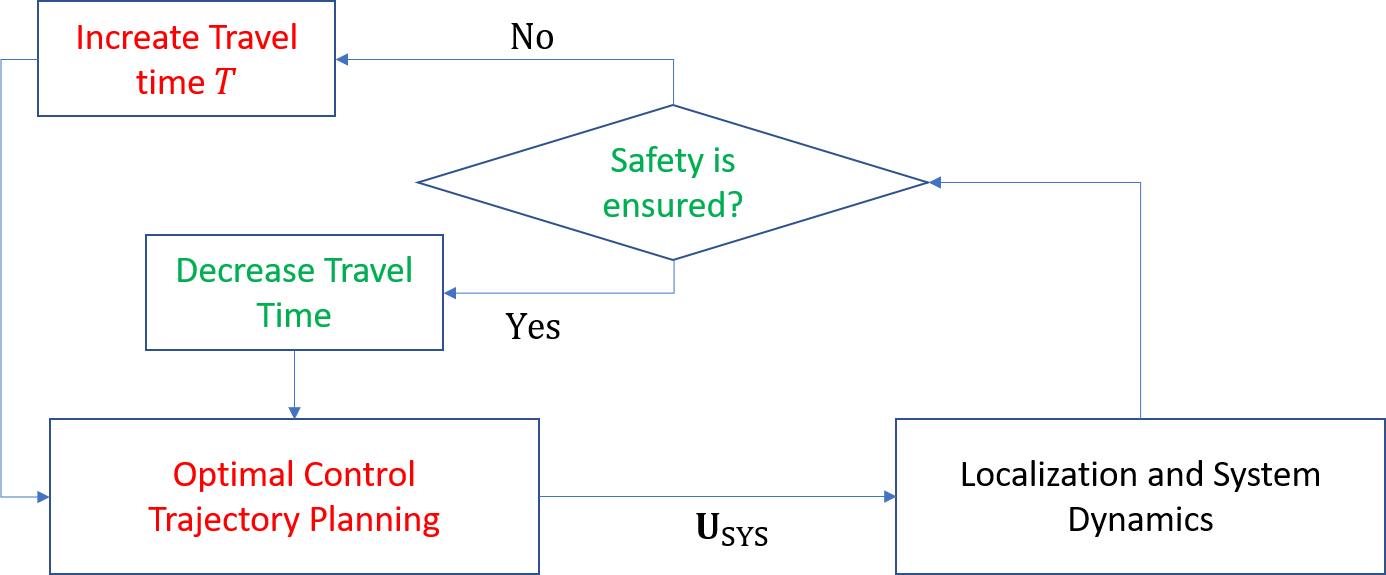}
\caption{The flowchart for assignment travel time $T$ ensuring safety of the MAS continuum deformation coordination.}
\label{SafetyFlowchart}
\end{figure}


\section{Safety Specification and Verification}\label{Safety Specification}
The following conditions provide safety requirements such as collision avoidance, boundedness, and follower containment in a continuum deformation coordination acquired by cooperative localization.

\textbf{Collision Avoidance Condition:} Let $\epsilon>0$ be the radius of the radius of the smallest ball enclosing every individual agent. 
Then, inter-agent collision is avoided, if
\begin{equation}
\label{const1111111}
\forall k,\qquad \bigwedge_{i\in \mathcal{V},~j\neq i} \|\mathbf{r}_i\left(k\right)-\mathbf{r}_j(k)\|>2\epsilon.
\end{equation}

\textbf{Boundedness Condition:} Deviation of every agent $i$ from the global desired trajectory $\mathbf{r}_{i,HT}(k)$ is bounded, if
\begin{equation}
\label{const222222222}
\forall k,\qquad     \bigwedge_{i\in \mathcal{V}}\|\mathbf{r}_i\left(k\right)-\mathbf{r}_{i,HT}(k)\|\leq\delta.
\end{equation}
where $\delta>0$ is constant.

\textbf{Follower Containment Condition:} Let $\mathbf{r}_m(k)=\left[x_m(k)~y_m(k)\right]^T$, $\mathbf{r}_j(k)=\left[x_j(k)~y_j(k)\right]^T$, $\mathbf{r}_h(k)=\left[x_h(k)~y_h(k)\right]^T$, $\mathbf{r}_t(k)=\left[x_t(k)~y_t(k)\right]^T$ be the position of arbitrary agents $m$, $j$, $h$, and $t$, respectively, where $m$, $j$, $h$ form a triangle at discrete time $k$, i.e. agents $m$, $j$, $h$ are not aligned at discrete time $k$. We define function 
\begin{equation}
    \mathbf{\Omega}\left(\mathbf{r}_m,\mathbf{r}_j,\mathbf{r}_h,\mathbf{r}_t\right)=\begin{bmatrix}
    x_m&x_j&x_h\\
    y_m&y_j&y_h\\
    1&1&1
    \end{bmatrix}
    ^{-1}
    \begin{bmatrix}
    x_t\\
    y_t\\
    1
    \end{bmatrix}
\end{equation}
at discrete every time $k$. Agent $t$ is inside the triangle defined by vertices $m$, $j$, and $h$, if $\mathbf{\Omega}\left(\mathbf{r}_m,\mathbf{r}_j,\mathbf{r}_h,\mathbf{r}_t\right)\geq \mathbf{0}$. We can ensured that all followers remain inside the leading triangle, defined by leaders $1$, $2$, and $3$, at every discrete time $k$, if
\begin{equation}
\label{const3333333333}
\forall k,\qquad     \bigwedge_{i\in \mathcal{V}}\mathbf{\Omega}\left(\mathbf{r}_1\left(k\right),\mathbf{r}_2(k),\mathbf{r}_3(k),\mathbf{r}_i(k)\right)\geq \mathbf{0},
\end{equation}
{\color{black}where ``$\bigwedge$'' is the ``wedge symbol.}
For given initial and final configurations of the leaders, we choose a sufficiently-large travel time $T=t_f-t_0\geq T^*$, where $T^*=t_f^*-t_0$ is assigned as the solution of the following constrained programming problem:
\begin{equation}
    T^*=\min T \label{min_T}
\end{equation}
subject to safety constraints \eqref{const1111111}, \eqref{const222222222}, \eqref{const3333333333}, MAS collective dynamics \label{canonicala}, estimation dynamics \label{canonicalb}, and the following inequality constraint:
\begin{equation}
    i\in \mathcal{V},~q\in \{x,y\},\qquad \ddot{q}_{i,HT}=v_{q,i}^*.
\end{equation}
where $q_{i,HT}^*$ is component $q\in \{x,y\}$ of desired trajectory of leader $i\in \mathcal{V}_L$. Note that $q_{i,HT}^*(t)$ can be determined, if initial and final configurations of the leaders are known, and $t_0$ and $t_f$. On the other hand, we need to know $q_{i,HT}^*$ of every leader agent $i\in \mathcal{V}_L$ to solve the minimum time optimization problem presented in this section. Therefore, leaders' desired trajectories and travel times must be solved interactively. The flowchart shown in Fig. \ref{SafetyFlowchart} illustrates how $T^*$ and $q_{i,HT}^*$ can be interactively determined for every leader $i\in \mathcal{V}_L$ such that all presented constraints are satisfied.

It must be noted that in EKF-based cooperative localization, random noise is introduced in both process and measurements. This leads to some uncertainty exist in the optimization process and results. Therefore, Monte Carlo simulations may be required to verify that all safety requirements are satisfied regardless of uncertainties in state estimation. 

\begin{figure}[ht]
\centering
\includegraphics[width=3.0   in]{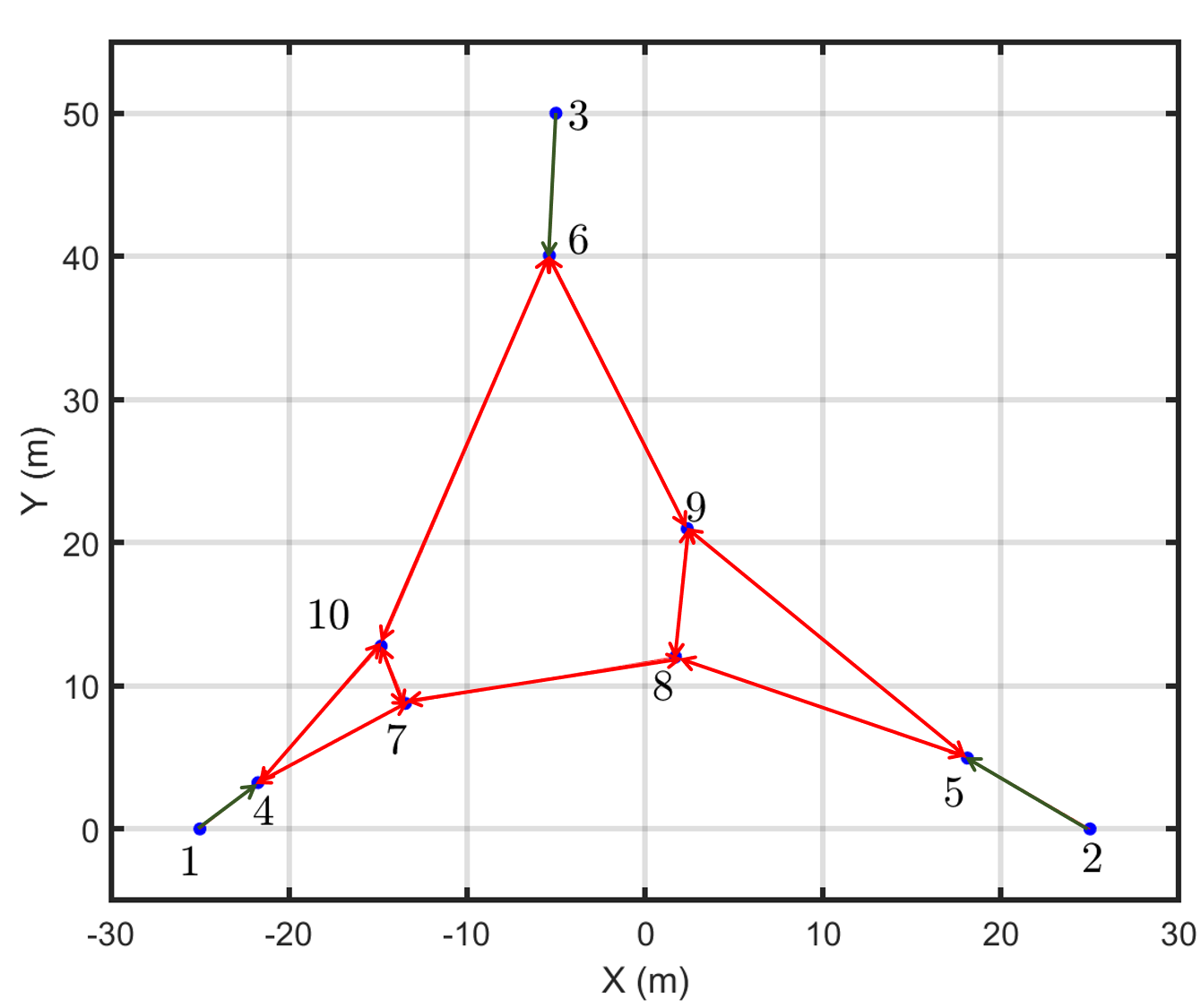}
\caption{Coordination graph.
}
\label{CoordinationGraph}
\end{figure}
\begin{figure}[ht]
\centering
\includegraphics[width=3.0   in]{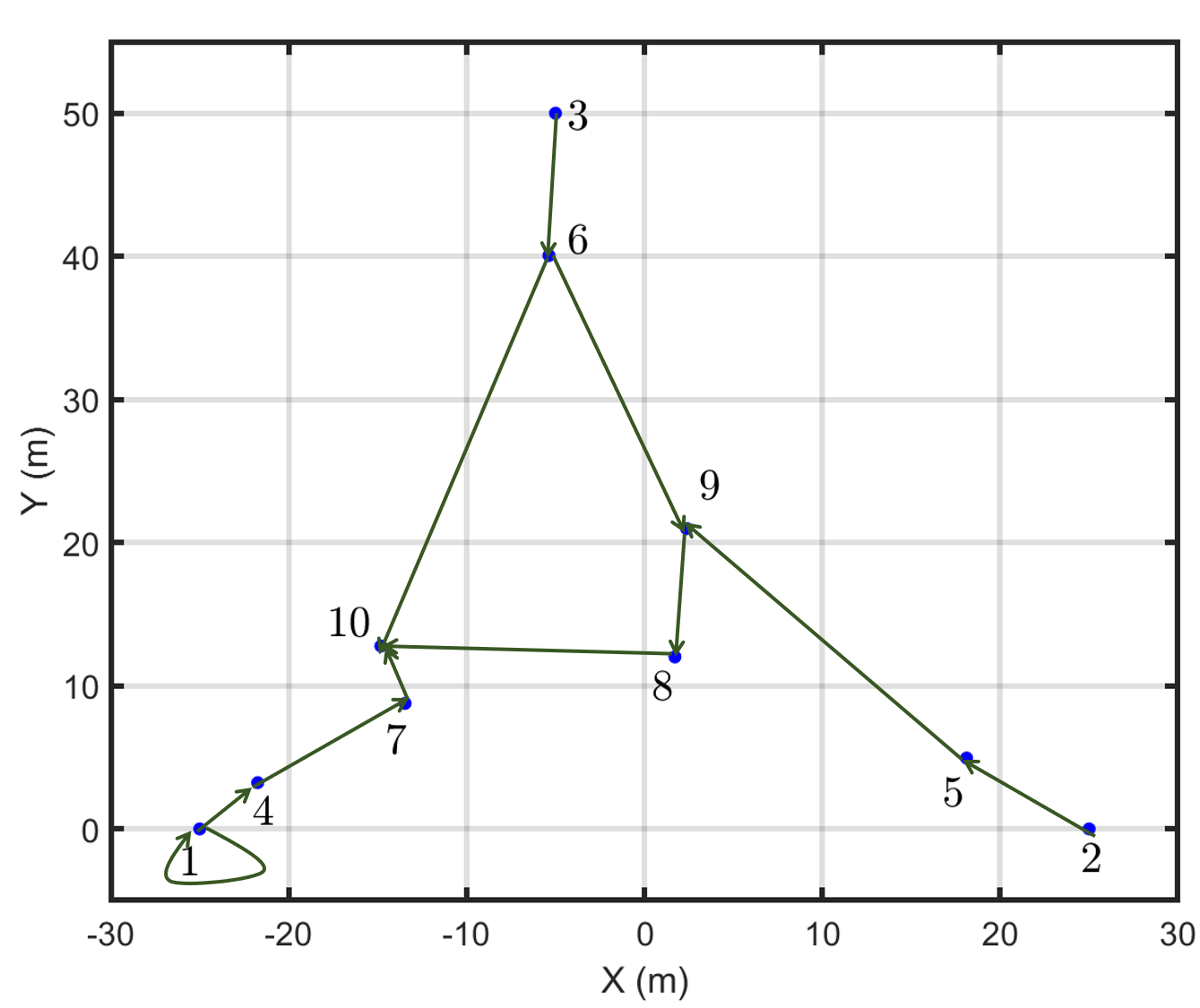}
\caption{Localization graph.}
\label{LocalizationGraph}
\end{figure}

\section{Simulation Results}\label{Simulation Results}
We consider continuum deformation of a MAS consisting of $10$ mobile robots identified by set $\mathcal{V}=\{1,\cdots,10\}$. The robot team moves in the $x-y$ plane where the MAS continuum deformation coordination is guided by leaders $\mathcal{V}_L=\left\{1,2,3\right\}$. The remaining mobile robots, identified by  $\mathcal{V}_F=\{4,\cdots,10\}$, are followers, and acquire the desired continuum deformation in a decentralized fashion through cooperative coordination and localization.

For cooperative coordination, the MAS uses the graph  $\mathcal{G}_c\left(\mathcal{V},\mathcal{E}_c\right)$ shown in Fig. \ref{CoordinationGraph} with the mobile robots located at their initial positions. Figure \ref{LocalizationGraph} shows is the graph $\mathcal{G}_l\left(\mathcal{V},\mathcal{E}_l\right)$ used in cooperative localization, which is mainly based on proximity. The graph implies that leader agent $1$ is capable of self measurement while agents $4,5,6,7,8,9$ can measure the relative positions of agents $1,2,3,4,9,5$. Meanwhile agent $10$ simultaneously measures the relative positions of agents $6$, $7$, and $8$.  

Since mobile robots are represented as double integrators, we must include process noise in the model to account for unmodeled dynamics. Hence we assume the standard deviation of the process noise to be $0.5~\text{m}/\text{s}^4$. We also assume that the measurements obtained by all agents are polluted by additive Gaussian noise with a standard deviation of $0.03~\text{m}$. In addition, all measurements are updated at a rate of $0.1~\text{s}$.

The minimum time $T^*$ determined from \eqref{min_T} which satisfies conditions \eqref{const1111111}, \eqref{const222222222}, and \eqref{const3333333333} for this scenario is determined to be only $10.1~\text{s}$ assuming continuous noise-free full-state feedback. However, the minimum time is increased to $20.5~\text{s}$ when states are estimated using cooperative localization, as determined by optimization and verified via Monte Carlo simulations. In these simulations, the gains for the control law in \eqref{main2ninput} were selected as $g_1=6$, $g_2=9$. The selected safety threshold values were $\epsilon=0.5$ and $\delta=0.5$.

Figure \ref{fig:paths} shows the paths of the 10 mobile robots in the simulation. Figure \ref{fig:estimation_errors} shows the global position estimation errors in $X$ and $Y$ directions for a typical simulation with cooperative localization. Figure \ref{fig:tracking_errors} shows the coordination tracking errors in $X$ and $Y$ directions. It can be concluded that the both estimation and tracking errors never exceed $0.8\%$ of the distance travelled by each agent.

\begin{figure}[ht!]
\centering
\includegraphics[width=3.3   in]{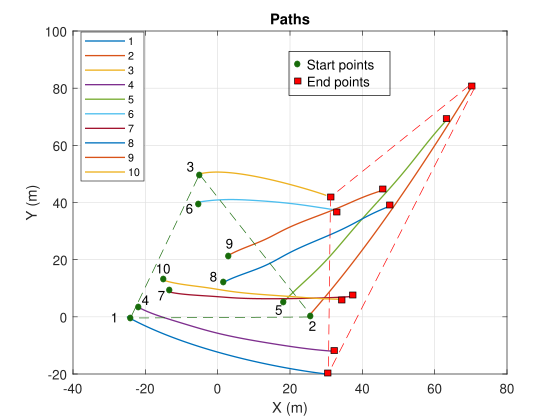}
\vspace{-0.3cm}
\caption{Paths of all 10 mobile robots.}
\label{fig:paths}
\end{figure}

\begin{figure}[ht!]
\centering
\includegraphics[width=3.3   in]{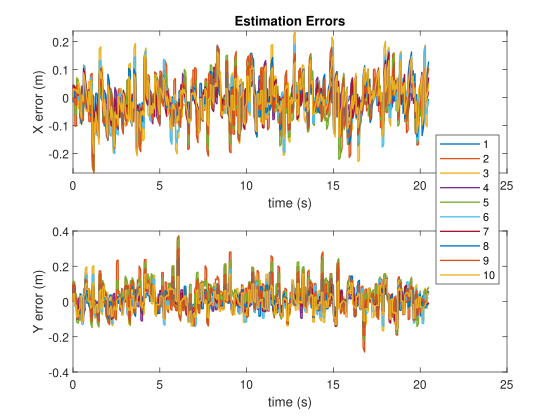}
\vspace{-0.3cm}
\caption{Position estimation errors of all 10 mobile robots.}
\label{fig:estimation_errors}
\end{figure}

\begin{figure}[ht!]
\centering
\includegraphics[width=3.3   in]{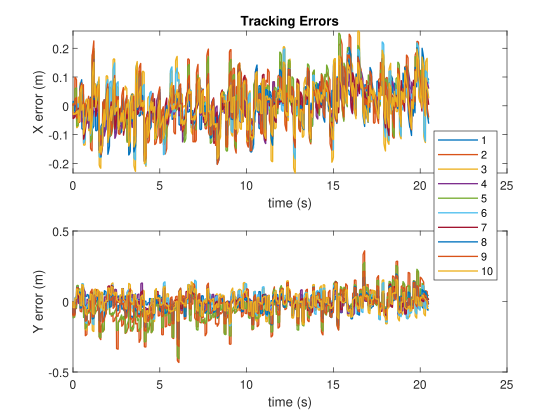}
\vspace{-0.3cm}
\caption{Tracking errors of all 10 mobile robots.}
\label{fig:tracking_errors}
\end{figure}

\section{Conclusion}\label{Conclusion}
This paper studied the problem of continuum deformation coordination using cooperative localization. The problem was presented as a decentralized leader follower coordination, where leaders' desired trajectories and travel time between the initial and final configurations are obtained by solving coupled optimization problem. More specifically, leaders' desired trajectories are determined by solving a constrained optimal control problem while the final time is minimized such that all safety requirements are satisfied. We also showed how follower agents can acquire the desired continuum deformation coordination through simultaneous communication and localization. As a result, a relatively-large number of agents of a multi-agent system can acquire a desired continuum deformation with low computation cost while all safety requirements are met and the MAS is capable of aggressive deformation in the geometrically-constrained environments.  
\section*{Acknowledgment}
This research was supported in part by the National Science Foundation under Award No. 1914581 and the Office of Naval Research under Award No. N00014-19-1-2255.




%



\bibliographystyle{IEEEtran}
\bibliography{reference}

\end{document}